\documentclass[sigconf]{aamas}  %

\usepackage{balance}  %
\newcommand{\leaveout}[1]{}

\usepackage{booktabs}

\setcopyright{ifaamas}  %
\acmDOI{}  %
\acmISBN{}  %
\acmConference[AAMAS'19]{Proc.\@ of the 18th International Conference on Autonomous Agents and Multiagent Systems (AAMAS 2019)}{May 13--17, 2019}{Montreal, Canada}{N.~Agmon, M.~E.~Taylor, E.~Elkind, M.~Veloso (eds.)}  %
\acmYear{2019}  %
\copyrightyear{2019}  %
\acmPrice{}  %

\usepackage{amsmath}
\usepackage{amssymb}
\usepackage{amsthm}
\usepackage{mathtools}
\usepackage{parskip}
\usepackage{physics}
\usepackage{algorithm}
\usepackage[noend]{algpseudocode}
\usepackage{bbm}
\usepackage[english]{babel}
\usepackage{color}

\newcommand{\iscabr}{\textsc{Issue Selection}}
\newcommand{\biscabr}{\textsc{BISC} }
\newcommand{\isc}{\textsc{Issue Selection Control} }

\allowdisplaybreaks

\begin{document}

\title{Manipulating Elections by Selecting Issues}

\author{Jasper Lu}
\author{David Kai Zhang}
\authornote{EECS, Vanderbilt University, USA}
\affiliation{}
\email{jasper.lu@vanderbilt.edu}
\email{david.k.zhang@vanderbilt.edu}

\author{Zinovi Rabinovich \and Svetlana Obraztsova}
\authornote{SCSE, Nanyang Technological University, Singapore}
\affiliation{}
\email{zinovi@ntu.edu.sg}
\email{lana@ntu.edu.sg}

\author{Yevgeniy Vorobeychik}
\authornote{CSE, Washington University in St. Louis, USA}
\affiliation{}
\email{yvorobeychik@wustl.com}

\begin{abstract}
Constructive election control considers the problem of an adversary
who seeks to sway the outcome of an electoral process in order to
ensure that their favored candidate wins. We consider the
computational problem of constructive election control via issue
selection. In this problem, a party decides which political issues to
focus on to ensure victory for the favored candidate.  We also
consider a variation in which the goal is to maximize the number of
voters supporting the favored candidate.  We present strong negative
results, showing, for example, that the latter problem is
inapproximable for any constant factor.  On the positive side, we show
that when issues are binary, the problem becomes tractable in several
cases, and admits a 2-approximation in the two-candidate case.
Finally, we develop integer programming and heuristic methods for
these problems.
\end{abstract}


\maketitle

\section{Introduction}
	
The study of the extent to which elections are susceptible to subversion by malicious parties has received considerable attention under the general framework of \emph{election control}.
The computational complexity of this problem has been formally studied from a number of perspectives, such as control by adding and deleting candidates and voters~\cite{Bartholdi92,Hemaspaandra07}, and in the context of different voting systems~\cite{Menton2013,faliszewski2009llull,erdelyi2010control,hemaspaandra2009hybrid,erdelyi2010parameterized}.
However, there is an important means of manipulating election outcomes that has been largely ignored: that of determining which \emph{issues} are discussed and, consequently, which are most salient for voters when they come to the polls.

To illustrate, take three issues,
healthcare, environmental regulation, and immigration, and
suppose that all voters want universal health coverage and
environmental regulation, and a slight majority wish to restrict
immigration.
Suppose that positions are binary (support or oppose).
Now, consider two candidates, one who supports immigration,
environmental regulation, universal
healthcare, and the second who is opposed to all three.
Clearly, if all issues are considered, the former
candidate wins in a landslide.
However, if one party is able to skew discourse
\emph{entirely} towards immigration, the second candidate may narrowly
win.

We investigate the problem of election control through manipulating issues (which can also be viewed as a novel variant of the \emph{bribery} problem~\cite{DBLP:reference/choice/FaliszewskiR16,Put2016,article}).
In this problem, we assume that voters and candidates can be represented as points in a vector space over issues, where each vector represents one's (voter's or candidate's) position on all issues, and the preference ranking of candidates by a voter is induced by the norm distance between their respective position on issues in the natural way.
We then investigate the election control problem in the context of a choice of a subset of issues, whereby the distance between a voter and a candidate in the resulting restricted issue space determines the relative standing of this candidate to others.
Our study considers several related variations of this general framework: the decision problem in which the interested party either aspires to have a candidate of their choice win, and the optimization problem of maximizing the support (total number of votes) for a target candidate, all in the context of plurality elections.

We obtain a series of strong negative results.
First, we show that not only is the general problem of controlling elections through manipulating issues NP-Hard for both the decision problem and the variant aiming to maximize support, it is actually inapproximable for any constant factor for the latter variant.
Moreover, the problem remains hard whether one breaks ties in favor of the target candidate, or not, and even when there is either a single voter, or two candidates.
Second, we show that the problem remains hard if we restrict issues to be binary. 
On the other hand, we observe that under certain restrictions we can obtain positive results.
For example, the problem is tractable if there is only a single voter (unlike in the general case), and maximizing support is 2-approximable when there are two candidates.
Finally, we provide solution approaches for these problems based on integer linear programming, as well as a greedy heuristic.

\subsection{Related Work} 
Our work is related to two areas of research on social choice: the spatial theory of elections (including lobbying) and election control.

\paragraph{Spatial Theory of Elections and Lobbying Models} 
	Spatial models of elections were first introduced by Hotelling \cite{Hotelling29}, with extensive research following in the decades since~\citep{Davis66,Enelow90,Mckelvey90,Merrill99,Anshelevich15,Anshelevich16,Skowron17}.
	A major focus areas of research in the spatial model of elections is that of a candidate choosing where to locate in a policy space \citep{Smithies41,Black1948}. A key development in this field is the \textit{Median Voter Theorem (MVT)} \citep{Black1948}, which characterizes the special case 
of our model with two candidates and one issue. 
In this case, the winning candidate is the one preferred by the median voter. 
However, MVT's assumptions of absolute candidate mobility and global attraction are unrealistic, which continues to stimulate research on this model~\citep{sorr_2017,DBLP:conf/atal/ShenW17}. 
Algorithmic work in the spatial model has been somewhat more sparse, although with several recent studies focusing largely on social choice functions and distortion relative to a natural social choice function caused by common voting rules, such as plurality~\citep{Anshelevich15,Anshelevich16,Skowron17}.

An important research area within the spatial model is lobbying, whereby an actor wishes to change decisions by voters on issues so that majority vote on each issue corresponds to this actor's preference~\cite{Christian07,Goldsmith14}.
The two clear difference from our proposed research is that in our case, issue preferences determine which candidate wins, rather than votes on each issue separately, and that in our case manipulation targets groups of voters, whereas lobbying research is typically focused on changing votes for a subset of $k$ voters.
Somewhat related research assumes voters and candidates are fixed actors in a policy space, and considers the game of convincing voters of a candidate's truthfulness \citep{AustenSmith90,Lott89}.

\paragraph{Election Control}  
 Election control research focuses on the problem of tampering with an election to either ensure that a candidate wins or loses an election. The spatial theory of elections aims to explain why voters vote the way they do by modeling an election system as sets of \textit{voters} and \textit{candidates} as positions in an $n$-dimensional policy space, in which voters vote for those candidates closest to them in Euclidean distance. 

The computational problem of constructive election control, in which an adversary manipulates an election to ensure that a candidate wins was first studied by Bartholdi et al. \cite{Bartholdi92}, while Hemaspaandra et al. \cite{Hemaspaandra07} initiated the study of destructive control. Much work since then has been done in election control under different voting systems, such as range voting \citep{Menton2013}, approval voting \citep{Erdelyi09}, and others \citep{faliszewski2009llull,erdelyi2010control,hemaspaandra2009hybrid,erdelyi2010parameterized}, as well as in bribery
 \citep{Faliszewski09,erdelyi2010control,article,Put2016}.

\section{Control through Issue Selection}

We study the problem of \emph{election control through issue selection}.
To do so, we impose structure on a voting problem by assuming that voter preferences over candidates are solely based on their relative stance on the issues.
To be more precise, consider a collection of $\ell$ issues, and a space $X \subseteq \mathbb{R}^\ell$ of possible positions on the issues.
Thus, $x \in X$ represents a vector of positions on all issues, with $x_k$ the position on (opinion about) issue $k$.
In our setting, we have a collection of $m$ candidates, $C = \{c_i\}^m_{i=1}$, and $n$ voters, $V = \{v_j\}^n_{j=1}$, where each candidate $i$ and voter $j$ is characterized by a position vector (representing their respective positions on all $\ell$ issues), which we denote by $c_i$ and $v_j$, respectively, with $c_i, v_j \in X$.
We use $c_{ik}$ ($v_{jk}$) to denote the position of candidate $i$ (voter $j$) on issue $k$, and we refer to the vector of a candidate's or voter's beliefs as their belief vector.
Denote by $[a:b]$ the interval of all natural numbers from $a$ to $b$, and suppose that voters consider a nonempty subset of issues, $S \subseteq [1:n], S \ne \emptyset$, in deciding which candidate to vote for.
This set $S$ captures those issues which are \emph{salient} to voters, for example, due to a focus on these during campaigning.
We assume that a voter $v_j$ will rank candidates in order of their relative agreement on issues, as captured by an $l_p$ norm for integral $p \ge 1$ with respect to the set of issues $S$.
Henceforth, we focus on plurality elections, so that a voter $v_j$ would vote for a candidate $i$ which minimizes $\|v_j^S - c_i^S\|_p$,
where $x^S$ denotes a restriction of $x$ to issues in $S$.

We consider two \emph{constructive control} problems within this framework: control through issue selection (\textsc{Issue Selection Control (ISC)}), and maximizing support (\textsc{Max Support}), which we now define formally.

\begin{definition}[\textsc{Issue Selection Control (ISC)}]
Given a set of candidates $C$, voters $V$, and $\ell$ issues, is there a nonempty subset of issues $S \subseteq [1:\ell]$ such that a target candidate $c_1$ wins the plurality election?
\end{definition}
\begin{definition}[\textsc{Max Support}]
Given a set of candidates $C$, voters $V$, and $\ell$ issues, find a nonempty subset of issues $S \subseteq [1:\ell]$ which maximizes the number of voters who vote for a target candidate $c_1$.
\end{definition}

For both problems, we must define a rule by which to break ties. 
We consider both the best-case of undecided voters choosing the target candidate $c_1$, and the worst-case of undecided voters choosing another candidate. We use the same tie-breaking rule when several candidates are tied.

\section{Real-Valued Issues}

We begin our study of election control by analyzing its algorithmic
hardness when issue positions are unrestricted, i.e., $X =
\mathbb{R}^\ell$.
We show that the problem is computationally intractable, even for %
a single voter or with only two candidates.
However, the problem is tractable when the number of issues is
bounded by a constant.

\subsection{Issue Selection with a Single Voter}

Consider election control through issue selection
with only a single voter, $v$, which we term \textsc{Single-Voter
  \iscabr\ (SVIS)}.
We start by assuming that ties are broken in candidate $c_1$'s favor
(best-case tie breaking).
Note that in this setting, \iscabr\  and \textsc{Max Support}
are essentially equivalent: in either case,
we ask whether there exists a nonempty subset of issues $S \subseteq [1:\ell]$ such that
when restricted to these issues, candidate $c_1$ wins the voter $v$
(with a maximum support of 1 if $c_1$ wins, and 0 if it loses). 
Equivalently, we ask if there exists a nonempty subset $S$ such that
\begin{equation} \label{svis_objective}
\sum_{k \in S} \abs{c_{1k} - v_k}^p \leq \sum_{k \in S} \abs{c_{ik} - v_k}^p \quad \text{$\forall i\in[2:m]$}
\end{equation}
where $v_k$ is the sole voter's position on issue $k$.
Observe that condition \eqref{svis_objective} holds if and only if
\[ \sum_{k \in S} \abs{c_{ik} - v_k}^p - \abs{c_{1k} - v_k}^p \geq 0 \quad \text{$\forall i\in[2:m]$} \]
Thus, setting the entries of an auxiliary $(m-1) \times \ell$ matrix $M$
\begin{equation}
M_{i-1,k} = \abs{c_{ik} - v_k}^p - \abs{c_{1k} - v_k}^p,
i\in[2:m], k\in[1:\ell]
\end{equation}
we can equivalently ask whether there exists a nonempty subset $S$ of
the columns of $M$ such that the restriction of $M$ to these has nonnegative row sums. We will refer to such a restriction of an election as "highlighting" a set of issues.

\begin{theorem}\label{thm:svis_bc_npc}
\textsc{SVIS} with best-case tie breaking is NP-complete for any $l_p$ norm.
\end{theorem}

\begin{proof} First observe that \textsc{SVIS} is in NP. Indeed, given an instance of \textsc{SVIS} and a proposed subset $S$, it is trivial to verify whether $S$ satisfies condition \eqref{svis_objective} in polynomial time.
	
	We now show that \textsc{SVIS} is NP-hard via reduction from
        \textsc{0-1 Integer Linear Programming}, which is well-known
        to be NP-complete. In this problem, we are given a matrix $A
        \in \mathbb{Z}^{m \times 
\ell}$ and a vector $b \in \mathbb{Z}^\ell$, and we ask if there exists a vector $x \in \{0, 1\}^\ell$ such that $Ax \geq b$ componentwise.
	
	Given an arbitrary instance $(A,b)$ of \textsc{0-1 Integer Linear Programming} (ILP), we construct an $(m+1) \times (\ell+1)$ matrix $M$ as follows:
	\begin{align*}
	M_{i,k} &\coloneqq A_{i,k} & i &= 1, \dots, \ell & \qquad k &= 1, \dots, \ell \\
	M_{i,\ell+1} &\coloneqq -b_i & i &= 1, \dots, m & & \\
	M_{m+1,k} &\coloneqq -\frac{1}{\ell+1} &&& k &= 1, \dots, \ell \\
	M_{m+1,\ell+1} &\coloneqq 1.
	\end{align*}
	This construction is motivated by the observation that
        choosing a subset $S$ of columns of $M$ so that $c_1$ wins the
        election is analogous to choosing the positions of ones in a vector $x$ that satisfies $Ax \ge b$. Each row of $M$ corresponds to a candidate belief vector with the constraint vector $b$ included as an added issue. We force this issue to be considered by creating a dummy candidate whose beliefs coincide with $c_1$ on all but that issue.
	
	We now construct an instance of \textsc{SVIS} by setting the voter belief vector $v$ to be the zero vector and constructing a sequence of candidate belief vectors $C = \{c_i\}_{i=2}^{m}$ from $M$.
	\begin{align*}
	c_{1k} &\coloneqq \sqrt[p]{\abs{\min_i M_{ik}}} & k \in[1:\ell+1]& \\
	c_{i+1,k} &\coloneqq \sqrt[p]{M_{ik} + c_{1k}^p} & i \in[1:m+1],&\ k\in[1:\ell+1]
	\end{align*}
	We do this because we want to arrange that $M_{ik} = \abs{c_{ik}}^p - \abs{c_{1k}}^p$, using positive values of $c_{ik}$ for simplicity.
	It is then straightforward to see that the original instance of \textsc{0-1 Integer Linear Programming} is satisfiable if and only if our constructed instance of \textsc{SVIS} is satisfiable, by constructing a 0-1 vector $x$ with ones at precisely the indices in $S \setminus \{\ell+1\}$, or vice versa.
\end{proof}

\begin{theorem}\label{thm:svis_wc_hard}
	The worst-case version of \textsc{SVIS} is at least as hard as the best-case version of \textsc{SVIS}.
\end{theorem}
\begin{proof}[Proof Sketch]
  Consider an $m \times \ell$ matrix $M$ representing an arbitrary
  instance of the best-case version of \textsc{SVIS} and
  define $$\epsilon = \min_{i \in [1:m], k \in [1:\ell]}
  \frac{1}{2}\abs{\sum_{k' \in R(k)}M_{i,k'}},$$
  where the set $R(k) = \{\binom{r}{k}, r \in [1..\ell]\}$.
        We can create a new $(m + 2 )\times (\ell + 1)$ matrix $M'$ as follows:
	\begin{align*}
		M'_{i,k} &\coloneqq M_{i,k} & i &= 1, \dots, m & k &= 1, \dots, \ell \\
		M'_{m+1,k} &\coloneqq 0 &  && k &= 1, \dots, \ell \\
		M'_{i,k+1} &\coloneqq \frac{\epsilon}{2} & i &= 1, \dots, m+ 1 & \\
		M'_{m+2,k} &\coloneqq \epsilon & & & k &= 1, \dots, \ell \\
		M'_{m+2,\ell+1} &\coloneqq -\frac{\epsilon}{2}. & & & 
	\end{align*}
	
	Recall that in the worst-case version of \textsc{SVIS}, a voter will default to other candidates in cases of a tie. 
	So, we are forced to include issue $\ell+1$ in $S$ in order to win against candidate $m + 1$. 
	Once we include issue $\ell+1$, we bias the voter towards the
        target candidate and against each candidate by a small amount.
	Because of our choice of $\epsilon$, this bias will only affect the election in instances where the candidates are tied.
	However, we still have to include at least one other issue from $[1:\ell]$ to win against candidate $m + 2$. 
	
	This construction then turns into the best-case version of
        \textsc{SVIS} once we begin to consider combinations of issues
        from $[1:\ell]$ with issue $m+1$.
\end{proof}

\subsection{Issue Selection with Two Candidates}

While issue selection is hard even with a single voter, we now ask
whether it remains hard if we have only two candidates.
We term the resulting restricted problem \textsc{Two-Candidate
  \iscabr\ 
 (TCIS)}.
We show that both of the considered problem variants remain NP-hard. Furthermore, \textsc{Max Support} is actually
inapproximable to any constant factor even in this restricted setting.

\begin{theorem}\label{thm:2c_isc_bc_npc}
	\textsc{TCIS} with best-case
        tie breaking is NP-complete.
\end{theorem}
\begin{proof}
	First, observe that \textsc{TCIS}
         is in NP because, given a set $S$ of issues to
        highlight, we can easily check if $c_1$ wins the election in
        polynomial time. We use a reduction from \textsc{0-1 Integer
          Linear Programming} to prove it's NP-Hard.

Next, consider the issue selection problem with two candidates and a set of
voters $V$.
Note that we successfully control the election iff the following
condition holds for at least half of the voters $v_j$ (remember that ties are broken in $c_1$'s favor):
\begin{equation} \label{fcis_objective}
\sum_{k \in S} \abs{c_{1k} - v_{jk}}^p \leq \sum_{k \in S}
\abs{c_{2k} - v_{jk}}^p %
\end{equation}
We now construct a matrix $M$ with entries 
\begin{equation}
M_{j,k} = \abs{c_{2k} - v_{jk}}^p - \abs{c_{1k} - v_{jk}}^p, j \in
[1:n], k \in [1:\ell].
\end{equation}
We can equivalently ask for a nonempty subset $S$ of columns of
$M$ such that the restriction of $M$ to those columns maximizes the
number of indices $j$ s.t.\ $\sum_{k \in S} M_{jk} \ge 0$. 
	
	Let $A$ be our ILP matrix, and $b$ - the ILP constraints. Then, we can reduce ILP to \textsc{TCIS} by creating the following $(2n + 1) \times (\ell+1)$ matrix $M$: %
	\begin{align*}
	M_{j,k} &\coloneqq A_{j,k} & j &\in[1:n]& k &\in[1:\ell]\\
	M_{j,k} &\coloneqq -1 & j &\in[n+1:2n+1] & k &\in[1:\ell]\\
	M_{j,\ell+1} &\coloneqq -b_j & j &\in[1:n]\\
	M_{j,\ell+1} &\coloneqq 0 & j &\in[n+1:2n]&\\
	M_{2n+1,\ell+1} &\coloneqq \ell + 1
	\end{align*}
	As in our reduction of \textsc{SVIS}, we represent the constraint vector $b$ as an
        issue that must be put in $S$ in order for $S$ to win. We also
        create $n$ dummy voters with all negative entries. This will
        force us to look for assignments of $S$ that satisfy all rows
        that correspond to constraints of ILP. If $c_1$ can win the
        given election, we return yes for ILP, and no if $c_1$ cannot
        win. 

        Finally, we show that for any $M$ we can derive voter preferences
consistent with it.
Since definition of $M$ is independent for different issues $k$, it
will suffice to do this for an arbitrary issue $k$ ($k$th column of
$M$, which we denote by $M_k$).
Consequently, consider a column $M^k$, and define $\bar{M}_k = \max_j
|M_{j,k}|$ (the value of $M_k$ with the largest magnitude).
Define $c_{1k} = 0$ and $c_{2k} = \bar{M}_k^{1/p}$.
Additionally, define a function $f(z) = |c_2 - z|^p - |c_1 - z|^p$ for
$z \in [0,c_2]$.
Clearly, this function is continuous, and $f(0) = \bar{M}_k$ while
$f(c_2) = -\bar{M}_k$.
Then by the intermediate value theorem, for any $M_{jk}$, we can find
a $v_{jk}$ such that $f(v_{jk}) = M_{jk}$.
Repeating the process for each issue $k$ gives us the construction.
\end{proof}

Next, we turn to the \textsc{Max Support} version of the issue
selection problem with two candidates; we term this
\textsc{Two-Candidate Max Support (TCMS)}.
We show that not only is it NP-hard, it is inapproximable.
\begin{theorem}
\label{T:2cms}
\textsc{TCMS} with best-case tie breaking is NP-hard for any $l_p$
norm. Moreover, it cannot be approximated to any constant factor
unless $P=NP$.
\end{theorem}
\begin{proof}
	We can now show that \textsc{TCMS} is NP-hard by restricting
        $\ell$ to $2$ and reducing from \textsc{Maximum Independent
          Set (MIS)}. Given an undirected graph $G = (V, E)$ on $|V|$ vertices, MIS asks to select a maximal subset of vertices $S \subseteq V$ so that $S$ is an independent set (i.e., no pair of vertices in $S$ is connected by an edge).%
	
	Given any instance of MIS, we can represent that instance as an instance of \textsc{TCMS} by first creating a $|V| \times |V|$ matrix with every value along the diagonal equal to $|V| - 1$. For every pair of vertices $u, v$, set $M_{u,v} = M_{v,u} \coloneqq -|V|$ if $u$ and $v$ are connected in $G$, and $-1$ otherwise. 
	Now, if we were to select an issue corresponding to vertex
        $u$ with neighbor $v$, then we cannot hope to select any other
        subset of issues such that row $v$ sums to greater than or
        equal to 0. Thus, the action of selecting columns of $M$ to
        include in $S$ corresponds to selecting vertices of $G$ to
        be in our independent set, and maximizing the number of rows
        in this manner corresponds to finding a maximum independent
        set. 

To complete the reduction, what remains 
to prove is that we can derive voter belief and candidate belief
vectors for any $M$ constructed in this manner.
The associated lemma is provided in the supplement.

Inapproximability 
follows directly from our reduction of \textsc{MIS} to
\textsc{TCMS}: we know that MIS is NP-hard to
approximate within any constant factor $c > 0$
\cite{Arora:2009:CCM:1540612}, and our reduction from MIS is
approximation-preserving.
\end{proof}

The next results show that the worst-case tie breaking setting is no
easier than when ties are broken in $c_1$'s favor.
\begin{theorem}\label{thm:isc_wc_hard}
	The worst-case version of \textsc{TCIS} is at least as hard as the best-case version of \textsc{TCIS} for the two-candidate case.
\end{theorem}
\begin{proof}[Proof Sketch]
Given an $n \times \ell$ matrix $M$ associated with a two-candidate
instance of best-case issue selection, define $\epsilon$ as in the
proof of Theorem~\ref{thm:svis_wc_hard}.
Further, we let $x\coloneqq \max\limits_{j\in[1:n],k\in[1:\ell]}\abs{M_{j,k}}$, and create a $3n \times (\ell+1)$ matrix $M'$ as follows:
\begin{align}
M'_{j,k} &\coloneqq M_{j,k} & j &\in[1:n]& k &\in[1:\ell]\nonumber\\
M'_{j,k} &\coloneqq x & j &\in[n+1:2n] & k &\in[1:\ell]\label{thm:isc_wc_hard_eq3}\\
M'_{j,k} &\coloneqq -x & j &\in[2n+1:3n]& k &\in[1:\ell]\label{thm:isc_wc_hard_eq4}\\
M'_{j,\ell+1} &\coloneqq \frac{\epsilon}{2} & j &\in[1:n]\nonumber\\
M'_{j,\ell+1} &\coloneqq -\frac{\epsilon}{2} & j &\in[n+1:3n]\nonumber
\end{align}

Once again, we choose a value of $\epsilon > 0$ such that $\epsilon$ will affect the election only if a voter is undecided. The proper assignment is shown in the supplement.

Recall that in the worst-case version of \textsc{TCIS}, undecided voters (rows of $M'$ with a net zero value) will default to a candidate other than $c_1$. 
With the addition of column $\ell+1$, any undecided voters will now be ``nudged'' in the direction of $c_1$ instead. 
Also, since the values of column $n+1$ are smaller than the difference of any two values of $M$, the issue affects the election only if a voter is actually undecided. 
So, issue $\ell+1$ appropriately mimics the weak inequality used in the best-case version of \textsc{TCIS}, and if a candidate wins an election in the worst-case reduction, they win the election in the best-case version, and vice versa. 

Note: we add $2n$ extra voters to the problem to set things up such that including issue $n+1$ would not be sufficient for winning the election. We also choose $2n$ voters specifically so that we can be guaranteed to split voters evenly between $c_1$ and $c_2$ with our assignments in Equations~\ref{thm:isc_wc_hard_eq3} and ~\ref{thm:isc_wc_hard_eq4}.
\end{proof}

\begin{corollary}
The worst-case version of \textsc{TCMS} is NP-hard.
\end{corollary}

\section{Binary Issues}

We have shown that election control through issue selection is hard in general.
However, real world opinions may have a variety of restrictions.
For example, legislative issues can be viewed as \emph{binary issues}, where a voter opinion can take only two values: support or oppose.

Formally, in binary versions of the issue selection problems, $X = \{0, 1\}^\ell$.
Voters vote for the candidate with whom they agree on most issues. 
Let \textsc{Binary Issue Selection Control (BISC)} be the variant of
\iscabr  over a binary domain and, similarly, let \textsc{Binary Max
  Support (BMS)} be the corresponding variant of the \textsc{Max Support} problem.

\subsection{Binary Issue Selection with 1, 2 and 3 Voters}

We start by considering again the problem of issue selection with a
single voter, which we showed to be NP-Hard in the general case of
real-valued issues.
We show that this problem is now in P.

As before, it suffices to consider solely \textsc{Single-Voter \biscabr}.
We start with the case when ties are broken in $c_1$'s favor (best-case tie-breaking).
Consider the following \textsc{Single Issue Win} algorithm:
\begin{algorithm}
Check if there is an issue such that either (a) $c_1$ agrees
  with the voter $v$, or (b) no other candidate $c_j$ agrees with $v$.
  If it exists, return YES.  Otherwise, return NO.
\end{algorithm}
\begin{theorem}
The \textsc{Single Issue Win} algorithm solves \textsc{Single-Voter \biscabr} with best-case tie-breaking.
\end{theorem}
\begin{proof}
It suffices to show that whenever \textsc{Single Issue Win} returns NO, $c_1$ cannot win the election.
Consider an arbitrary subset of issues $S$.
Since the answer is NO, it must be that for each issue $k \in S$, $c_1$ disagrees with $v$ on $k$.
Consequently, $\|v - c_1\| = |S|$.
Choose a $c_j$ which agrees with $v$ on some issue $k \in S$.
Then $\|v - c_j\| \le |S|-1$, that is, $c_1$ cannot win for issues restricted to $S$.
Since $S$ is arbitrary, the result follows.
\end{proof}

In fact, we can easily generalize the algorithm for a single voter to a
setting with two voters by simply applying the algorithm for each voter.
  \begin{corollary}
\label{C:bsic-bc}
    \textsc{2-Voter \biscabr} problem with best-case tie-breaking is poly-time solvable.
  \end{corollary}

Next, we show that the problem is in P for one and two voters even with worst-case
tie-breaking, although the algorithmic approach is quite different.
For worst-case tie-breaking, we propose the following \textsc{Agree On Issues} algorithm:
\begin{algorithm}
Let $S_{agree}$ be the set of all issues on which $c_1$ agrees with $v$.
If $c_1$ wins over each other candidate $c_j$ when issues are restricted to $S_{agree}$, return YES. Otherwise, return NO.
\end{algorithm}
\begin{theorem}\label{thm:aoi_algo_correct}
The \textsc{Agree On Issues} algorithm solves \textsc{Single-Voter \biscabr} with worst-case tie-breaking.
\end{theorem}
\begin{proof}
It suffices to consider the case when we return NO.
Suppose there is some $c_j$ that wins when we restrict to $S_{agree}$.
Then it must be that $c_j$ also agrees with $v$ on all issues in
$S_{agree}$ (and any subset thereof).
Consider an arbitrary subset of issues $S$, and let $x_{jk} = 1$ if
$j$ agrees with $v$ on issue $k$.
$c_j$'s difference from $v$ is then $\sum_{k \in S \cap S_{agree}}
x_{jk} + \sum_{k \in S - S\cap  S_{agree}} x_{jk} \ge |S \cap
S_{agree}|$.
Since the difference between $c_i$ and $v$ is $|S \cap
S_{agree}|$, the result follows.
\end{proof}

  The same approach is also applicable to \textsc{2-Voter BIS}.
  \begin{corollary}\label{2v-bis-wc-poly}
    \textsc{2-Voter \biscabr} with the wost-case tie-breaking is poly-time solvable.
  \end{corollary}
  \begin{proof}
    For the candidate $c_1$ to win, {\bf both} voters must support
    her. Without loss of generality, we can assume that $c_1$ opinion on all isues is $1$. Let $S_{agree}$ be the set of all issues on which $c_1$ agrees with both voter $v_1$ {\bf and} $v_2$. Similarly to Theorem~\ref{thm:aoi_algo_correct}, if $c_1$ does not win against each other candidate $c_j$ over the set $S_{agree}$, then no other subset of issues will achieve $c_1$'s win. 
  \end{proof}

Remarkably, 
while BSIC with 1 and 2 voters are efficiently solvable for both
best-case and worst-case tie-breaking, with
3 voters we see a qualitative difference in complexity, depending on
how ties are broken.
First, we observe that the 3-voter case with worst-case tie-breaking is tractable.
  \begin{corollary}\label{thm:3v-bis-wc-poly}
    \textsc{3-Voter Binary Issue Selection} with the worst-case tie-breaking is poly-time solvable.
  \end{corollary}
  \begin{proof}
   By Corollary~\ref{C:bsic-bc} we can test in poly-time whether any
   given pair of voters can be won over by $c_1$.
Applying this to each of the three possible pairs of voters, we can determine in poly-time whether the support of any two voters can be obtained simultaneously. If so, then $c_1$ can be made to win. Otherwise no subset of issues will make $c_1$ the winner. 
  \end{proof}
Now, we show that the problem becomes hard with best-case
tie-breaking even with only 3 voters.
  \begin{theorem}\label{thm:3v-bis-bc-npc}
    \textsc{3-Voter Binary Issue Selection} with the best-case tie-breaking is NP-hard.
  \end{theorem}
  \begin{proof}
    The proof relies on a reduction from the \textsc{Exact 3-Cover (X3C)} problem. An instance of X3C is governed by $t$ -- number of elements, $s$ -- the number of sets. In the reduced instance we will denote by $w$ the preferred candidate (and assume that his opinion on all issues is $1$), $c$ -- the candidate whose opinion on every issue is $0$ (zero), $v_3$ -- the voter whose opinion on every issue is $0$. This implies that to win the election $w$ should gain the support of both voters $v_1$ and $v_2$. In addition we will denote by $r$ the number of issues in the reduced instance, setting it to $r=s+t+2$. Finally, we will set the number of candidates to $m=t+4$ and name them $c_1,\dots,c_t,x,y,c,w$.

The preferences of $v_1$ and $v_2$ over the $r$ issues are as follows:
$\begin{array}{rcccc}
  v_1: &1\dots 1&0\dots 0&1&0\\
  v_2: &\underbrace{0\dots 0}_{s}&\underbrace{1\dots 1}_{t}&0&1
\end{array}
$

Preferences of candidates take a more complex form
\begin{itemize}
\item \emph{For issues from $1$ through $s$}. These preferences will encode the X3C instance. In particular, candidates $c_i,c_j,c_e$ will have opinion $1$ on the $k^{th}$ issue if and only if the $k^{th}$ set in the X3C instance is $\{i,j,e\}=S_k$. Otherwise the opinion of these three candidates on the $k^{th}$ issue will be $0$ (zero).
\item \emph{On issues $s+1$ through $s+t$} all candidates $c_1,\dots,c_t$ have $0$ (zero) opinion.
\item \emph{On the $s+t+1$ issue} all candidates $c_1,\dots,c_t$ have opinion $0$ (zero)
\item \emph{On $s+t+2$ issue} all candidates $c_1,\dots,c_t$ have opinion $1$
\item \emph{Candidate $y$} has opinion $1$ on issues $1,\dots,s+t$ and opinion $0$ (zero) on the issues $s+t+1$ and $s+t+2$
\item \emph{Candidates $x$} has opinions in the complete opposion to candidate $y$
\end{itemize}

Let us now show that if we have a solution to the resulting \isc problem, we can recover a solution for the original X3C instance.

Candidate $c$, with all his opinions set to $0$ (zero), serves as a kind of reference for voters. Thus, given a selection $S$ of issues, the preferred candidate $w$ will gain the support of a voter only if they agree on at least as many issues in $S$ as they disagree. As a result, \iscabr  solution should contain equal number, $q$, of issues from the set $\{1,\dots,s,s+t+1\}$ and from the set $\{s+1,\dots,s+t,s+t+2\}$. Consequently, candidate $w$ will agree with any voter on exactly $q$ issues.

Notice that both the issue $s+t+1$ and $s+t+2$ must be selected in a solution to the $\isc$. To see this consider the follwoing two cases
\begin{itemize}
\item \emph{Neither $s+t+1$, nor $s+t+2$ are in the solution set, $S$ of issues}. Still, an equal number of elements (denoted earlier by $q$) must be selected from the sets of issues $\{1,\dots,s\}$ and $\{s+1,\dots,s+t\}$ for the solution set $S$. Wlog., issue $1\in S$. Then voter $v_1$ agreed with the candidate $c_{i_1}$ on $q+1$ issues ($q$ issues from the set $\{s+1,\dots,s+t\}$ and issue $1$). As a result, voter $v_1$ would \emph{not} vote for candidate $w$. Thus $S$, that does not contain neither $s+t+1$ nor $s+t+2$, can not be a valid solution to our \biscabr instance.
\item \emph{Only one among issues $s+t+1$ and $s+t+2$ is selected as a part of the solution set of issues $S$.} If it is the issue $s+t+1$, then voter $v_1$ agreed with the candidate $x$ on $q+1$ issues and with candidate $w$ on $q$ issues only. Thus, $v_1$ would not vote for $w$, and $S$ is not a valid soluion. Similarly, if $s+t+2$ was selected, then candidate $y$ will win the support of $v_2$, once again preventing $w$ from winning.
\end{itemize}

Now, with both issues $s+t+1$ and $s+t+2$ chosen, let us show how we can obtain a solution to the original X3C problem from the solution set of issues $S$ to the reduced \biscabr problem. The set of issues $S$ makes candidate $w$ the winner of the election. Let $\{i_1,\dots,i_{q-1}\}=S\cap \{1,\dots,s\}$. We will show that the collection $S_{i_1},\dots,S_{i_{q-1}}$ is the solution to the original X3C instance.

\begin{enumerate}
\item If \emph{there is an element $j$ that belongs to two different sets in the collection $S_{i_1},\dots,S_{i_{q-1}}$,} then $v_1$ agrees with $c_j$ on 2 issues from $i_1,\dots,i_{q-1}$ and on $q-1$ issues from $\{s+1,\dots,s+t\}$. Totalling $q+1$ agreements between $v_1$ and $c_j$. Which implies that $v_1$ will \emph{not} vote for $w$, and contradicts $w$ being the winner.
\item If \emph{there exists an element $j$ that does not belong to any set in the collection $S_{i_1},\dots,S_{i_{q-1}}$,} then $c_j\in C\setminus\{c_{j_i},c_{k_i},c_{e_i}\}$ for all $i\in\{i_1,\dots,i_{q-1}\}$. As a consequence $v_2$ agrees with $c_j$ on $q-1$ issues from the set of issues $\{1,\dots,s\}$ and on both issues $s+t+1$ and $s+t+2$. This totals $q+1$ agreements between $v_2$ and $c_j$, entailing that $v_2$ will \emph{not} vote for $w$, contradicting $w$ being the winner.
\end{enumerate}
As a result, the collection $S_{i_1},\dots,S_{i_{q-1}}$ constructed from the \biscabr solution $S$ is a proper solution to the original X3C instance, i.e. every element belong to 1 and only 1 set.

Let us now show that a solution to the X3C instance can be translated into a solution to the \biscabr reduction instance. 

Let $S_{i_1},\dots,S_{i_k}$ be a legal solution to the X3C instance. Then set the selection of issues $S=\{i_1,\dots,i_k\}\cup\{s+1,\dots,s+k\}\cup\{s+t+1,s+t+2\}$. Notice that $k$ is the number of elements in the X3C instance, and therefore $k=\frac{t}{3}$ and $s+k<s+t$.

By the choice of $i_1,\dots,i_k$, it must hold that $v_1$ agrees with every candidate $c_j$ once on issues $i_1,\dots,i_k$ and $\frac{t}{3}$ times on issues $s+1,\dots,s+k,s+t+1,s+t+2$. Overall $v_1$ and $c_j$ agree on $\frac{t}{3}+1$ issues. Candidate $x$ agrees with $v_1$ on issues $s+1,\dots,s+k,s+t+1$ only, totalling $\frac{t}{3}+1$ agreements as well. Similarly, candidates $c$ and $y$ rake in $\frac{t}{3}+1$ agreements. Thus, by the tie-breaking rule, $v_1$ votes for $w$.

Similarly, $v_2$ is matched with the opinion of $c_j$ over $\frac{t}{3}-1$ issues from the set $\{i_1,\dots,i_k\}$ and 2 more matches are produced over issues $s+t+1, s+t+2$. This totals $\frac{t}{3}+1$ matches between $c_j$ and $v_2$. Similarly to $v_1$, $v_2$ also agrees with $x,y$ and $c$ on $\frac{t}{3}+1$ issues. Again, tie-breaking will decide in favour of $w$.
Thus $w$ has the support of both $v_1$ and $v_2$ and becomes the winner.

We conclude that the original X3C instance has a solution if and only if the reduction instance of \biscabr has a solution.
    \end{proof}

\subsection{Binary Issue Selection with Two Candidates}

With an arbitrary number of voters and only two candidates, even the
\biscabr problem with best-case tie-breaking is hard.
\begin{theorem}\label{thm:bis-bc-npc}
	With two candidates, \biscabr with best-case tie-breaking is NP-complete.
\end{theorem}
\begin{proof}
It is evident that \biscabr problem is in NP, so we only need to show that it is NP-hard. We will do so by a reduction from \textsc{Hitting Set}, where $p$ denotes the number of elements, $s$ -- the number of sets, and $k$ --- the number of elements which should be chosen as the hitting set. We construct a profile for \biscabr problem with 2 candidates, $\ell$ issues and $n$ voters, where $\ell$ is such that $\ell=p+k$ and $n=2ks+4$.

We assume that the preferred candidate is $c_1$ and set his opinion to $1$ on all issues. All opinions of his rival, $c_2$, are set to $0$ (zero). We then arrange voters into $3$ blocks, as follows:

\begin{itemize}
\item {\bf [Block 1.]} Two voters. 
The first one has opinion $0$ (zero) for issues from $1$ through issue $\ell-k$, and opinion $1$ for issues from $\ell-k+1$ to $\ell$. The second voter has an opposite opinion wrt all issues.
\item {\bf [Block 2.]} Second block consists of $ks$ voters divided
  into k sub-blocks.
For every sub-block, opinions of voters on issues from $1$ to $\ell-k$ encode the hitting set problem instance. That is, voter $(f-1)s+i$ has opinion 1 on issue $j$ if and only if element $j \in s_i$ for all $f \in [1:k]$ . For issues from $\ell-k+1$ to $\ell$, all voters of the sub-block $f\in [1:k]$ will have the same $0$ (zero) opinion on issue $\ell-k+f$ and $1$ on all other issues. %
\item {\bf [Block 3.]} This block consists of $ks+2$ voters whose opinion on all issues is 0.
\end{itemize}
\leaveout{
\begin{table*}[ht]
{\center
\begin{tabular}{|cccc|cccc|cccc|ccc|cccc|}
\hline
\multicolumn{4}{|c|}{}&
\multicolumn{4}{c|}{{\bf Sub-block 1}}&
\multicolumn{4}{c|}{{\bf Sub-block 2}}&
\multicolumn{3}{c|}{{\bf ...}}&
\multicolumn{4}{c|}{{\bf Sub-block $k$}}\\\hline

\multicolumn{4}{|c|}{voters/issues}&
$1$&$2$&$\dots$&$s$&
$s+1$&$s+2$&$\dots$&$2s$&
\multicolumn{3}{c|}{$\dots$}&
$(k-1)s+1$&$(k-1)s+2$&$\dots$&$ks$\\\hline

\multicolumn{4}{|c|}{1}&
\multicolumn{4}{c|}{Hitting}&
\multicolumn{4}{c|}{Hitting}&
\multicolumn{3}{c|}{$\dots$}&
\multicolumn{4}{c|}{Hitting}\\

\multicolumn{4}{|c|}{2}&
\multicolumn{4}{c|}{Set}&
\multicolumn{4}{c|}{Set}&
\multicolumn{3}{c|}{$\dots$}&
\multicolumn{4}{c|}{Set}\\

\multicolumn{4}{|c|}{$\vdots$}&
\multicolumn{4}{c|}{Problem}&
\multicolumn{4}{c|}{Problem}&
\multicolumn{3}{c|}{$\dots$}&
\multicolumn{4}{c|}{Problem}\\

\multicolumn{4}{|c|}{$\ell-k$}&
\multicolumn{4}{c|}{Encoding}&
\multicolumn{4}{c|}{Encoding}&
\multicolumn{3}{c|}{$\dots$}&
\multicolumn{4}{c|}{Encoding}\\

\multicolumn{4}{|c|}{$\ell-k+1$}&
$0$&$0$&$\dots$&$0$&
$1$&$1$&$\dots$&$1$&
\multicolumn{3}{c|}{$\dots$}&
$1$&$1$&$\dots$&$1$\\

\multicolumn{4}{|c|}{$\ell-k+2$}&
$1$&$1$&$\dots$&$1$&
$0$&$0$&$\dots$&$0$&
\multicolumn{3}{c|}{$\dots$}&
$1$&$1$&$\dots$&$1$\\

\multicolumn{4}{|c|}{$\vdots$}&
$\vdots$&$\vdots$&$\dots$&$\vdots$&
$\vdots$&$\vdots$&$\dots$&$\vdots$&
\multicolumn{3}{c|}{$\dots$}&
$\vdots$&$\vdots$&$\dots$&$\vdots$\\

\multicolumn{4}{|c|}{$\ell$}&
$1$&$1$&$\dots$&$1$&
$1$&$1$&$\dots$&$1$&
\multicolumn{3}{c|}{$\dots$}&
$0$&$0$&$\dots$&$0$\\\hline
\end{tabular}\\
}
\caption{Theorem~\ref{thm:bis-bc-npc}. Voters' profile construction. Block 2.\label{thm:bis-bc-npc:prf:t1}}
\end{table*}
}

Let us now show the correctness of this reduction. Let $\{i_1,\dots,i_j\}$ be a set issues chosen to make $c_1$ the winner. Consider voters who support $c_1$. Evidently, nobody from Block 3 is among them --- no matter which issues were chosen, voters from Block 3 will support $c_2$. As a result, $c_2$ has at least $ks+2$ votes. Hence, all voters from Blocks 1\&2 should vote for $c_1$ to make him the winner.

Consider voters in Block 1. They both vote for $c_1$, therefore, $\{i_1,\dots,i_j\}$ consists of equal number of elements from both issue sets $[1:\ell-k]$ and $[\ell-k+1:\ell]$. Otherwise, there are (w.l.o.g.) more issues from $[1:\ell-k]$ than from $[\ell-k+1:\ell]$. Which implies that the second voter from Block-1 has more negative ($0$) opinions than positive ($1$), and he will vote for the candidate $c_2$. Additionally that means at most $k$ issues were picked from both sets. Denote this number by $r \leq k$.

W.l.o.g. issue $\ell-k+1$ is chosen from the set $[\ell-k+1:\ell]$. Thus, voters from the first sub-block of Block-2 have $r-1$ $1$'s and one $0$ as an opinion on issues in $[\ell-k+1:\ell]$. Therefore, all voters from this sub-block should have at least one positive ($1$) opinion on issues chosen from the issues set $[1:\ell-k]$. That is, these issues represent a hitting set with $r$ elements where $r \leq k$.

Similarly a solution for the \biscabr can be constructed from a given \textsc{Hitting Set} solution.

This proof is easy to adapt to worst-case tie-breaking.
\end{proof}
\begin{corollary}
	The \textsc{BMS} problem is NP-hard.
\end{corollary}

Although \textsc{Binary Max Support} is NP-hard, we now show that
it is easy to achieve a $\frac{1}{2}$-approximation using the
following \textsc{Best-Single-Issue} algorithm: choose one issue that
maximizes the net number of voters $c_1$ captures.
\begin{theorem}
	The \textsc{Best-Single-Issue} algorithm approximates $2$-candidate \textsc{Binary Max Support} to within a factor of $\frac{1}{2}$, for best-case and worst-case tie-breaking.
\end{theorem}
\begin{proof}
Let's denote number of voters by $n$ and the number of issues by
$\ell$. Among two candidates $c_1$ and $c_2$ the promoted one is
$c_1$. Without loss of generality, we can assume that candidate $c_1$ has opinion 1 on every issue. We will provide proof for the case of best-case tie-breaking and will describe changes needed to transform this proof into proof for worst-case tie-breaking. 

\begin{itemize}
\item {\bf [Case 0.]} There is an issue s.t. candidate $c_2$ also has opinion 1 about this issue. Therefore, if we highlight only this issue all voters will vote for $c_1$ because of tie-breaking. Thus, that is an optimal solution (and as such approximation within factor 2 of optimal solution). It is also the issue that captured the greatest number of voters for $c_1$ if highlighted. From now on we can assume that opinion of candidate $c_2$ is 0 for all issues.

\item {\bf [Case 1.]} There exists an issue s.t. at at least $\frac{n}{2}$ voters have same opinion as $c_1$. If highlighted such issue will capture for $c_1$ at least $\frac{n}{2}$ voters. That is, for issue that causes $c_1$ to capture the greatest number of voters it is at least $\frac{n}{2}$ voters too. Thus, it is provide $\frac{1}{2}$-approximation, because optimal solution is at most $n$.

\item {\bf [Case 2.]} Now we can assume that for all issues less then $\frac{n}{2}$ voters have opinion 1. Denote the largest such number by $h$ and show that $h$ is $\frac{1}{2}$-approximation of optimum. Assume the contrary. W.l.o.g. issues $s_1, \dots, s_k$ maximizes support for candidate $c_1$. By choice of $h$ the number of opinions which equals to 1 over all issues $s_1, \dots, s_k$ is at most $kh$. On the other hand voter supports candidate $c_1$ if and only if he has opinion 1 for at least $\frac{k}{2}$ issues among $s_1, \dots, s_k$. By assumption there are strictly more than $2h$ such issues. That is, on issues $s_1, \dots, s_k$ opinion 1 shared strictly more than $\frac{k}{2}2h=kh$ times. Obtained contradiction proves the theorem.
\end{itemize}

This proof can be easily adopted for the case of worst-case tie breaking. 
It is easy to see that if candidate $c_2$ has opinion 1 on all issues then for every highlighted set of issues support of candidate $c_1$ will be 0. Thus, any single issue provides $\frac{1}{2}$-approximation of optimum. Therefore, we may assume that there exist issue on which candidate $c_2$ has opinion 0.

Evidently, if there is optimum $s_{i_1}, \dots, s_{i_k}$ such that on
some of highlighted issues candidate $c_2$ has opinion
1. W.l.o.g. this issue $s_{i_k}$ then $s_{i_1}, \dots, s_{i_{k-1}}$ is
also optimum. Therefore, we may assume that candidates have different
opinions on all issues. Thus, it is enough to consider cases 1 and
2. The proof for case 1 remains unchanged. For case 2 we should change the
counting of number of points needed to obtain at least $2h$ votes in
favor of candidate $c_1$. A voter would only vote for $c_1$ if he has opinion 1 for $\left\lfloor \frac{k}{2}\right\rfloor +1$ 
issues $s_1, \dots, s_k$. Therefore, the number of opinions 1 is
$(\left\lfloor \frac{k}{2}\right\rfloor +1)2h \geq kh+h$, yielding same contradiction as in best-case tie-breaking. 
\end{proof}

\section{Algorithmic Approaches}
\label{S:algo}

We now present several general algorithmic approaches for 
\textsc{Max Support}: 1) exact approaches
based on integer linear programming (ILP), and 2)
a heuristic approach which works well in practice.

\noindent{\bf Integer Linear Programming: }
Define $A$ as follows:
\begin{equation}
A_{ijk} = \abs{c_{ik} - v_{jk}}^p - \abs{c_{1k} - v_{jk}}^p, 
\forall i \in [2:m], j \in V, k \in [1:\ell].
\end{equation}
Define $\alpha \coloneqq \sum_{ijk} \abs{A_{ijk}}$. 
The following ILP computes an optimal solution for (best-case)
\textsc{Max Support}:
\begin{subequations}
\begin{align}
& \max_x \sum_i^{m} y_i &\\ 
& \sum_k A_{ijk}x_k + (1-y_j)\alpha \ge 0 && \forall i \in [2:m], j \in V
&\label{C:1beatsi}\\
& x_k, y_j \in \{0, 1\} && \forall k \in [1:\ell], j \in V.%
\end{align}
\end{subequations}
Constraint~\eqref{C:1beatsi}, ensures that $y_j = 1$ iff $c_1$ is the
most favored by voter $j$. A similar approach can be used to develop a ILP approach for the
\textsc{Issue Selection Control} problem.

\noindent{\bf Greedy Heuristic: }
Finally, we present a simple greedy algorithm for the \textsc{Max
  Support} problem, where we iteratively add one issue at a time that
maximizes the net gain in voters.
We stop by adding any single issue would decrease the number of voters captured.

\section{Experiments}
	We now compare the performance of our exact and heuristic solution algorithms for the binary and continuous versions of the issue selection problem.
	We consider the greedy heuristics described above, as well as \textsc{Best-Single-Issue}.
	
	We run all of our experiments assuming a worst-case tie-breaking rule and generate random synthetic test cases.
	For continuous test problems, we sample candidate and voter belief vectors from the multivariate normal distribution with a mean of $0$ and a random covariance matrix.
    A similar generative model for Boolean issues, tends to produce problem instances in which \textsc{Best-Single-Issue} is nearly always optimal.
        Consequently, we generate a more specialized distribution of these instances as follows.
	We first construct a vertex-weighted complete binary tree $T$ on $2^\ell - 1$ vertices.
	Each vertex $v$ is assigned an independent random weight $p_v$ drawn from the uniform distribution on $[0,1]$.
	To produce a sample from $T$, we perform a directed random walk from its root to one of its leaves.
	The sequence ($0$ for left movements, and $1$ for right) emitted by this process is then the desired sample from $\{0,1\}^\ell$.

	We default to $3$ candidates, $100$ voters, and $10$ issues. 
	To generate each plot, we fix $2$ of these parameters and vary the $3$rd.
	We generate $100$ instances of \textsc{Max Support} for each set of parameter values, and run the heuristics on the instances.
	The plotted values are averages of the ratio of the number of voters captured %
and the optimal solution.

	\begin{figure}[h!]
		\begin{center}
		\begin{tabular}{cc}
		\includegraphics[width=40mm]{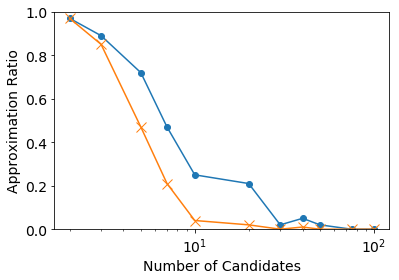} &
		\includegraphics[width=40mm]{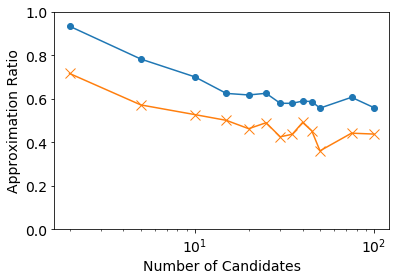} \\
		\includegraphics[width=40mm]{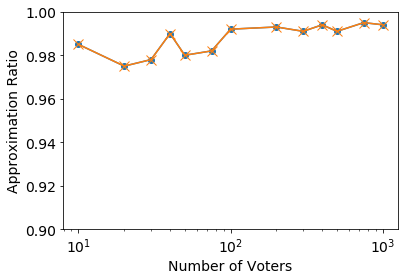} &
		\includegraphics[width=40mm]{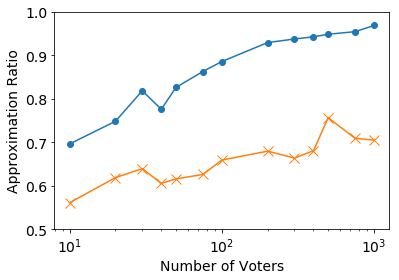} \\
		\includegraphics[width=40mm]{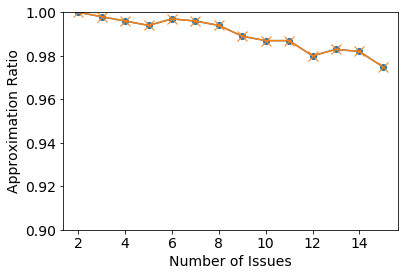} &
		\includegraphics[width=40mm]{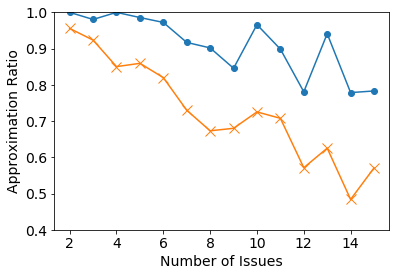} \\
		\multicolumn{2}{c}{\includegraphics[width=50mm]{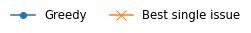}}
		\end{tabular}
		\end{center}
		\caption{Plots of experimentally observed approximation ratios as functions of the numbers of candidates, voters, and issues in synthetic test cases for binary (left) and continuous (right) versions of \textsc{Max Support}.}
	\end{figure}

We find that for most instances of \textsc{Max Support} with binary
issues, our greedy heuristic does not significantly outperform
\textsc{Best-Single-Issue} in the two-candidate setting as number of
issues and voters increase.  This is because the number of instances
in which a combination of issues can get us more voters than a single
best issue is increasingly unlikely.  However, the greedy algorithm
outperforms \textsc{Best-Single-Issue} on instances of \textsc{Binary
  Max Support} with greater than $2$ candidates.  We can also observe
that on the specific distribution of binary issue instances we
generate, the quality of heuristic solutions degrades rapidly with the
number of candidates.
	
	We find that for \textsc{Max Support} with real-valued issues, the greedy algorithm significantly outperforms \textsc{Best-Single-Issue}.
	For a small number of candidates ($< 5$), the greedy algorithm seems to perform within $0.8$ of optimal.
	Interestingly, as the number of voters increases, the greedy algorithm improves in quality on our randomly generated problem instances.
	In all cases, we can also observe that the heuristics tend to be close to optimal.

\section{Conclusion}
When candidates participate in an election, they must choose policies
and issues to stress in their campaigns. 
We introduce and study the problem of election control through issue selection. 
We find a number of strong negative results for the problem, and show
that, even though we canno t provide formal approximation guarantees
for a continuous instance of \textsc{Max Support}, a simple greedy
heuristic performs well.
Moreover, restricting issues to be binary admits further positive
results, including a $1/2$-approximation.

\begin{acks}
  This work was in part supported by NSF (IIS-1526860, IIS-1905558), NTU SUG M4081985 and MOE AcRF-T1-RG23/18 grants.
\end{acks}

\bibliographystyle{ACM-Reference-Format}  %
\balance  %

\end{document}